\documentclass[11pt]{article}
\usepackage{amssymb}
\usepackage{amsmath}
\usepackage{mathdots} 
\usepackage[capitalise,nameinlink]{cleveref}
\usepackage{tikz}
\usetikzlibrary{shapes.geometric,backgrounds,calc}
\usepackage{turnstile} 

\newcommand{\IN}{\mathbb N}

\newcommand\ma[1]{{\cal#1}}

\newcommand\pda{\textsc{pda}}
\newcommand\pdas{\textsc{pda}s}

\newcommand\oca{\textsc{oca}}
\newcommand\ocas{\textsc{oca}s}

\newcommand{\strong}{{\sf strong}}
\newcommand{\accept}{{\sf accept}}
\newcommand{\weak}{{\sf weak}}

\newcommand{\GENERA}[1]{%
  \mathrel{\vbox{\offinterlineskip\ialign{%
    \hfil##\hfil\cr
    $\scriptscriptstyle{#1}$\cr
    \noalign{\kern.4ex}
    $\Rightarrow$\cr
}}}}

\newcommand{\Mreaches}[1]{\sststile{\cal M}{#1}}
\newcommand{\Preaches}[1]{\sststile{{\cal M}'}{#1}}
\newcommand{\Nreaches}[1]{\sststile{\cal N}{#1}}

\newcommand*{\qed}{\mbox{}\nolinebreak\hfill~\raisebox{0.77ex}[0ex]{\framebox[1ex][l]{}}}

\newtheorem{theorem}{Theorem}

\newtheorem{lemma}{Lemma}
\newtheorem{corollary}{Corollary}

\newtheorem{definition}{Definition}
\newenvironment{proof}{\noindent{\em Proof.}}{\bigskip\noindent}

\newtheorem{ex}{Example}
\newenvironment{example}{\begin{ex}\rm}{\end{ex}}
\newenvironment{example-cont}[1]{\bigskip\noindent\textbf{Example~\ref{#1}.~(cont.)\hspace{\labelsep}}}{\bigskip\noindent}

\newcommand*{\eq}{{{\sf Eq}}}
\newcommand*{\meq}{{\ensuremath{\ma{A_{\eq}}}}}
\newcommand*{\meqs}{{\ensuremath{\ma{A_{\eq^*}}}}}
\newcommand{\bin}[1]{{{\sf bin}\ensuremath{(#1)}}}
\newcommand{\last}[1]{{{\sf last}\ensuremath{(#1)}}}
\newcommand{\ListBin}{{\ensuremath{\sf ListBin}}}
\newcommand{\Lsq}{{\ensuremath{\sf L_{sq}}}}
\newcommand{\Lk}[1]{{\ensuremath{{\sf L}}_{#1}}}
\newcommand{\Uk}[1]{{\ensuremath{{\sf U}}_{#1}}}
\newcommand{\Ustar}{{\ensuremath{\sf U_*}}}

\newcommand{\Valid}[1]{{\mbox{\sf Valid$(#1)$}}}
\newcommand{\pNotValid}[1]{{\mbox{\sf pNotValid$(#1)$}}}

\newcommand{\occurr}[2]{|#1|_{#2}}
\newcommand{\Log}{\lg}
\newcommand{\logk}[1]{\Log^{(#1)}\!}
\newcommand{\logstar}{\Log^*\!}

\newcommand{\tetration}[2]{{#1}^{\uparrow\uparrow#2}}
\newcommand{\tetra}[1]{{\tetration{2}{#1}}}

\newcommand{\Ext}[1]{{\mbox{\sf Ext$(#1)$}}}
\newcommand{\ExtK}[2]{{\mbox{\sf Ext$^{#1}(#2)$}}}

\begin{document}
\title{Turn Complexity of Context-free Languages, Pushdown Automata and One-Counter Automata\footnote{A preliminary version of this paper appeared in the proceedings of the conference DLT 2025~\cite{Pig25}.}}
\author{Giovanni Pighizzini\addtocounter{footnote}{2}\footnote{The author is a member of the Gruppo Nazionale Calcolo Scientifico-Istituto Nazionale di Alta Matematica (GNCS-INdAM).}\\
{\normalsize Dipartimento di Informatica}\\{\normalsize Universit\`{a} degli Studi di Milano, Italy}\\
{\normalsize\tt pighizzini@di.unimi.it}
}
\date{}
\maketitle
\noindent

%
%

%
%

\begin{abstract}
\noindent
A turn in a computation of a pushdown automaton is a switch from a phase in which the height of the pushdown store increases to a phase in which it decreases.
Given a pushdown or one-counter automaton, 
we consider, for each string in its language, the \emph{minimum number of turns made in accepting computations}.
We prove that it cannot be decided if this number is bounded by any constants. Furthermore, we obtain a non-recursive trade-off
between pushdown and one-counter automata accepting in a finite number of turns and finite-turn pushdown automata, that are defined requiring that the
constant bound is satisfied by \emph{each accepting computation}.
We prove that there are languages accepted in a sublinear but not constant number of turns, with respect to the input length.
Furthermore, there exists an infinite proper hierarchy of complexity classes, with the number of turns bounded by different sublinear functions.
In addition, there is a language requiring a number of turns which is not constant but grows slower than each of the functions defining the above hierarchy.
\end{abstract}


\section{Introduction}
Recently, we investigated two complexity measures related to context-free languages, pushdown automata, and one-counter automata.
The \emph{height of the pushdown} measures the space used in the pushdown store~\cite{PP23}.
In the restricted case of one-counter automata, in which the pushdown is used only to represent a counter written in unary, this measure
corresponds to the maximum value reached by the counter~\cite{PP25}. The \emph{push complexity} takes into account the number of push
operations used in computations~\cite{Pig2026}.
It should be clear that when one of these measures is bounded by a constant, namely it does not depend on the input length, then
the number of possible configurations of the pushdown store is finite. Hence, the store can be encoded in the finite control.
This implies the regularity of the accepted language.

The non-constant case is more interesting. First of all, there are fundamental differences if the costs of
all computations are taken into account, using the so-called \strong\ measure, or only
the less expensive accepting computation is considered for each accepted input (\weak\ measure).
In the \strong\ measure, if one of the above-mentioned resources is not bounded by any constant, then an amount that linearly grows
with respect to the length of the input is required~\cite{PP25,Pig2026}.
Similar results hold if the costs of \emph{all accepting} computations (instead of those of all computations) are taken 
into account (\accept\ measure)~\cite{PP25}.

The situation for the \weak\ measure is more complicated, but more challenging.
We point out that this measure is strictly related to the idea on nondeterminism in a broad sense: a nondeterministic machine is not only able
to make choices that lead to acceptance, if an accepting strategy does exist, but it is also able of taking choices that lead to
the \emph{less expensive} accepting strategy.
For the two resources mentioned at the beginning, it has been proved that if a machine uses an amount which is not bounded by any
constants, then this amount should grow at least as a double logarithmic function in the length of the input.
This lower bound is optimal for input alphabets of at least two symbols, while different and higher bounds (between
a logarithmic and a linear amount, depending on the resource) have been obtained in the unary case.

In this paper, we continue this line of investigation by considering a further resource, namely the number of \emph{turns} made
in computations.
Roughly, a pushdown automaton makes a turn during a computation, when a sequence of moves in which the height of
the store increases is followed by a sequence in which the height decreases.
In a pioneering paper on formal languages and pushdown automata, Ginsburg and Spanier started the investigation of
the role of turns in pushdown automata computations~\cite{GS66}.
In particular, they introduced \emph{finite-turn pushdown automata}, namely pushdown automata for which there exists a constant
limiting the number of turns made in \emph{each accepting computation}. The class of languages accepted by these devices
is a proper subclass of context-free languages and properly contains the class of
regular languages. Even restricting to one-counter automata that makes only one turn, the resulting class of languages is wider than
regular languages. In particular, 1-turn pushdown automata characterize the class of \emph{linear context-free languages}.

These facts already reveal a main difference between the number of turns and resources as pushdown height and push complexity:
one turn is enough to have more power than only with a finite control.
In the same paper, Ginsburg and Spanier presented several interesting results. Among them, they proved that it is decidable 
if a pushdown automaton is finite turn and if it is~$k$-turn, for a given~$k$. They also obtained a characterization of languages accepted
by finite-turn pushdown automata in terms of a special kind of context-free grammars, called \emph{ultralinear grammars}.
For this reason, the class of languages accepted by finite-turn automata is called class of \emph{ultralinear languages}.
We point out that \emph{all accepting} computations of the machine are taken into account, namely the \accept\ measure is considered.

It is quite natural to ask what happens by considering the \weak\ measure, i.e., if
for each accepted input, \emph{only} an accepting computation using the lowest number of turns is taken into account.
In the following, when we say that a machine \emph{accepts in} a certain number of turns, we always refer to the \weak\ measure.
Ginsburg and Rice observed that when a machine accepts in a constant number~$k$ of turns, namely
for each accepted input there is an accepting computation containing at most~$k$ turns, then the machine can be modified in order
to stop when the number of turns exceeds~$k$. The resulting \emph{$k$-turn pushdown automaton} is equivalent to the original one.
So, even under the \weak\ measure, the class of languages accepted in a finite number of turns coincides with
the class of ultralinear languages.
However, in contrast with the fact that it can be decided whether a pushdown automaton is finite turn (\accept\ measure), 
Ginsburg and Rice also proved that it is not decidable
whether a pushdown automaton accepts an ultralinear language~\cite{GS66}.

Going deeper into this investigation, in the first part of this paper, after presenting in Section~\ref{sec:prel} preliminary notions and a lemma that will be play a crucial role in several of the proofs contained in the paper,  in Section~\ref{sec:undec} we prove that it is  undecidable whether a pushdown automaton accepts in a finite number of turns (\weak\ measure).\footnote{We point out that this problem is different from the above mentioned problem of deciding if a pushdown automaton accepts a ultralinear language. Indeed, we can build pushdown automata accepting ultralinear languages, or even a regular languages,  but using an unbounded number of turns.}
Actually we prove that the property is undecidable already in the restricted case of one-counter automata. 
We also prove that it cannot be decided if a
one-counter automaton accepts in~$k$ turns, for a given~$k>0$. Only acceptance in~$0$ turns is decidable
for these devices as well as for pushdown automata.

We mentioned that each pushdown automaton accepting in~$k$ turns can be modified to obtain an equivalent~$k$-turn
pushdown automaton. However, as we prove in the paper, the value~$k$ cannot be bounded by any recursive function in the size of
the given machine. As a consequence, we are able to prove that the size trade-off between pushdown automata accepting in a 
constant number of turns and finite-turn pushdown automata is nonrecursive.\footnote{For a discussion on nonrecursive trade-offs we address the
reader to~\cite{Kut05}.}\,\footnote{%
Further results concerning the role of turns in the computations of pushdown automata, in both nondeterministic and deterministic cases,
have been obtained by Malcher~\cite{Mal07} and further expanded in~\cite{MP13}.}

\smallskip

The second part of the paper is devoted to study how the number of turns can grow with respect to the length of the input,
still considering the \weak\ measure.
It is easy to give examples of languages accepted in a constant number of turns and of languages accepted in a linear
number of turns. Here, we show the existence of languages accepted in a sublinear number of turns, for several
sublinear functions.
First, in Section~\ref{sec:sublinear} we prove that there is a language accepted by a one-counter automaton in~$O(\sqrt n)$
turns, where~$n$ is the length of the input. 
Furthermore, the same language cannot be accepted in~$o(\sqrt[3]n)$ turns, even using a pushdown automaton.

We then prove that there are machines and languages using even a lower number of turns.
In Section~\ref{sec:hie} we present a construction that, given a pushdown (resp., one-counter) automaton accepting a language in~$t(n)$ turns, allows to
obtain a pushdown (resp., one-counter) automaton accepting a different language in~$t(\log n)$ turns.
By iterating this construction, we are able to arbitrarily reduce the number of turns.
In particular, with a refinement of this construction, for each integer~$k$ we prove the existence of a language~$\Lk{k}$ accepted by a one-counter automaton
in~$\logk{k}n$ turns but that cannot be accepted
by any pushdown automaton in a number of turns growing less 
than~$\logk{k}n$.\footnote{$\logk{k}$ denotes the composition of the base~$2$ logarithm with itself~$k$ times.}
This result reveals a big difference with pushdown height and push complexity, for which there are double logarithmic lower bounds
when they are not constant, and defines a proper infinite hierarchy of complexity classes with respect to the number of turns.

It is possible to go further below: in Section~\ref{sec:ustar} we show the existence of a language~$\Ustar$ which is accepted by
a one-counter automaton in~$\logstar n$ turns (the \emph{iterated logarithm} of~$n$, a function growing slower than~$\logk{k}n$, 
for any~$k$). We also prove that such a number is necessary.

\medskip

We conclude the paper by presenting some final remarks in Section~\ref{sec:conclu}.

\section{Preliminaries}
\label{sec:prel}

We assume that the reader is familiar with the standard notions of automata and formal language theory as presented in textbooks, e.g.,~\cite{HU79}.
Given an alphabet~$\Sigma$, the set of strings over~$\Sigma$ is denoted by~$\Sigma^*$, with the empty
string denoted by~$\varepsilon$. The length of a string~$x\in\Sigma^*$ is denoted as~$|x|$, while the number of occurrences of a symbol~$a\in\Sigma$ in~$x$ as~$\occurr{x}{a}$. 
Given a set~$S$, its cardinality is denoted as~$\# S$, the family of
its subsets as~$2^S$, its complement as~$S^c$.

Let~$\log_2n$ denote the logarithm in base~$2$. 
For technical reasons, we modify this function in order to have a \emph{nonnegative} range.
More precisely we consider the function~$\Log n$, defined for~$n\geq 0$ as follows:
\[
\Log n=
\left\{\begin{array}{lll}
	\log_2n	&~~~&\mbox{if } n\geq 1,\\
	0 & &\mbox{otherwise}.
\end{array}\right.
\]
Given an integer~$k\geq 0$, $\logk{k}$ denotes the function obtained by composing~$k$ times the
function~$\Log$ with itself, namely
$\logk{k}n = \underbrace{\Log \Log \cdots \Log}_{k~\textrm{times}}n$.
We point\vspace*{-0.5cm}\linebreak[4]out that~$\logk{0}$ is the identity function.

\smallskip

The \emph{iterated logarithm} is the function~$\logstar:\IN\rightarrow\IN$ defined as
\[\logstar n= \min\{k\mid\logk{k}n\leq 1\}\,.\]
This function  grows very slowly. For instance~$\logstar n = 4$ for~$16 < n \leq 65536$,
while~$\logstar n = 5$ for~$65536 < n \leq 2^{65536}$.

Given integers~$x,k\geq 0$, let us denote by~$\tetration{x}{k}$ the number~$\underbrace{x^{x^{{\iddots}^{x}}}}_{k~{\rm times}}$, i.e.,
\[
\tetration{x}{k}=\left\{\begin{array}{ll}
	1 & \mbox{if } k=0,\\
	x^{(\tetration{x}{k-1})} & \mbox{otherwise}.
\end{array}\right.
\]
It is not difficult to observe that~$\logk{j}(\tetration{2}{k})=\tetration{2}{k-j}$, for~$0\leq j\leq k$, and hence~$\logk{k}(\tetration{2}{k})=1$.

The following inequality will be used in the paper:

\begin{lemma}
\label{lemma:logk-prod}
	Let~$\alpha>\beta>1$, then~$\Log(\alpha\beta)< 2\log\alpha$, and for~$k\geq 2$:
	\[\logk{k}(\alpha\beta) < 1+\logk{k}\alpha\,.\]
\end{lemma}
\begin{proof}
	First we observe that~$\Log(\alpha\beta)<\Log\alpha^2=2\Log\alpha$.
	
	For~$k=2$, we have:
	\[
	\logk{2}(\alpha\beta) = \Log\Log(\alpha\beta) < \Log\Log\alpha^2 = \Log (2\Log\alpha) = 1 + \logk{2}\alpha\,.
	\]
	We then consider~$k>2$. By induction we obtain:
	\[
	\logk{k}(\alpha\beta) = \log\logk{k-1}(\alpha\beta) < \log(1+\logk{k-1}\alpha) \leq 1 + \logk{k}\alpha\,,
	\]
	where, in the last step, we used the fact that~$\Log(1+x)\leq 1 +\Log x$, for~$x\geq 1$.
\qed
\end{proof}

\smallskip

We shortly remember the notion of pushdown automata as presented in~\cite{GS66,HU79}.
A \emph{pushdown automaton} (\pda, for short) is defined as~$\ma{M}=(Q,\Sigma,\Gamma,\delta,q_0,Z_0)$,
where~$Q$, $\Sigma$, and $\Gamma$ are finite sets: the set of states, the input alphabet, the pushdown alphabet, respectively;
$q_0\in Q$ is the initial state; $Z_0\in\Gamma$ is the start symbol on the pushdown store;
$\delta$ is the \emph{transition function} from~$Q\times(\Sigma\cup\{\varepsilon\})\times\Gamma$ to the finite subsets of~$Q\times\Gamma^*$,
where~$(p,\alpha)\in\delta(q,\sigma,A)$, with~$p,q\in Q$, $\sigma\in\Sigma\cup\{\varepsilon\}$, $A\in\Gamma$, $\alpha\in\Gamma^*$,
means that the automaton in the state~$q$, with~$A$ on the top of the pushdown store, reading the symbol~$\sigma\in\Sigma$ from
the tape, or without reading any symbol when~$\sigma=\varepsilon$, can move in the state~$p$ after replacing~$A$ on the top of the
stack with the string~$\alpha$.\footnote{%
	We use the convention that pushdown strings are written from top to bottom; hence, when~$\alpha\neq\varepsilon$, the leftmost symbol of~$\alpha$ will be
	on the top of the stack after this transition.
}
We assume \emph{acceptance by empty store}, namely an input~$w$ is accepted when there exists a computation starting from
the initial configuration  (input head on the first symbol of~$w$, finite control in the state~$q_0$, pushdown store containing only~$Z_0$)
that, according to the transition function, reaches a configuration in which all the input has been scanned (i.e., the input head is to 
the right of the rightmost symbol of~$w$) and the pushdown store is empty.
It is well known that the class of language accepted by these devices coincides with the class of 
\emph{context-free languages}. 

A \emph{turn} in a computation~$\cal C$ of a \pda\ is a sequence of~$k\geq 2$ moves
such that the height of the pushdown increases in the first move, decreases in the last move,
and it does not change in between, namely a turn is defined by~$k+1$ consecutive configurations, in which the
pushdown contents are~$\gamma_0, \gamma_1, \ldots, \gamma_k$, and~$|\gamma_0|<|\gamma_1|=|\gamma_2|=\cdots=|\gamma_{k-1}|>|\gamma_k|$.

According to~\cite{GS66}, a \pda~$\ma{M}$ is said to be~\emph{$k$-turn}, for an integer~$k\geq 0$, when \emph{each accepting computation} of~$\ma{M}$ contains at
most~$k$ turns.  A \pda~$\ma{M}$ is \emph{finite turn} if it is~$k$-turn for some~$k\geq 0$.
It is easy to see that~$0$-turns \pdas\ recognize only regular languages.
Furthermore, the class of languages accepted by~$1$-turn \pdas\ is the class of~\emph{linear languages}. 
This class is strictly included in the class of languages accepted by finite turns \pdas, which is the class of 
\emph{ultralinear languages}, a proper subclass of context-free languages.

In the above definitions, the number of turns in \emph{each accepting computation} is taken into account.
Since in a nondeterministic machine several computations using different numbers of turns could be possible, we are going to introduce
a different notion. In particular, we are going to consider the number of turns
which are \emph{sufficient} to accept input strings. More precisely, we give the following definition:

\begin{definition}
\label{def:turns}
	Let~$\ma{M}$ be a \pda\ accepting a language~$L$.
	We say that \emph{a string~$w\in L$ is accepted in~$k$ turns}, for an integer~$k\geq 0$, 
	when there exists a computation of~$\ma{M}$ accepting~$w$ and containing at most~$k$ turns.
	
	Given a function~$t:\IN\rightarrow\IN$, we say that \emph{$\ma{M}$ accepts~$L$ in~$t(n)$ turns} when, for each~$w\in L$
	with~$|w|=n$, there is an accepting computation on~$w$ containing at most~$t(n)$ turns. 
	Furthermore, if the function~$t(n)$ is bounded by a constant, then we say that \emph{$\ma{M}$ accepts in a finite number of turns}.
\end{definition}

If we know that a \pda\ accepts in~$k$ turns, for some given~$k\geq 0$, then we can always obtain an equivalent~$k$-turn \pda\ just
counting in the finite control the number of turns and stopping the computation when such a number exceeds~$k$.
However, as we will see, acceptance in~$k$ turns cannot be decided, unless~$k=0$.

According to~\cite{Har78}, the size of a \pda\ is the number of symbols used to write down its description and it
depends on the cardinalities of its alphabets and set of states, and of the maximum number of pushdown symbols that appear on the
right-hand sides of transition rules.

\medskip

A \emph{one-counter automaton} (\oca, for short) is a pushdown automaton in which the pushdown alphabet
consists only of one symbol~$A$, besides~$Z_0$. The symbol~$Z_0$ is used only to mark the bottom. 
In this way, the pushdown store always contains a string of the form~$A^kZ_0$, representing the integer~$k\geq 0$, with
the exception of accepting configurations, in which the store is empty.
It is well known that the class of languages recognized by \ocas\ is a proper subclass of context-free languages and 
properly contains the class of regular languages.

\medskip

In the paper, we use the language~$\eq=\{a^nb^n\mid n> 0\}$, which is a linear context-free language
accepted by a trivial deterministic \oca~$\meq$ that makes exactly one turn in each accepting computation.
By~$\meqs$ we denote the trivial deterministic \oca\ accepting the Kleene closure~$\eq^*$ of~$\eq$ and that makes exactly~$k$ turns to accept each string in~$\eq^k$, $k>0$.
It is folklore that~$k$ turns are necessary and sufficient to accept~$\eq^k$, for any fixed~$k\geq 0$.
Furthermore, each \pda\ accepting~$\eq^*$ should make~$k$ turns to accept any string of the 
form~$a^{n_1}b^{n_1}a^{n_2}b^{n_2}\cdots a^{n_k}b^{n_k}$, for sufficiently large~$n_1,n_2,\ldots,n_k$.

In Lemma~\ref{lemma:lowerBoundEQ}, we present a generalization of these lower bounds that will be used several times in the paper.

To prove it, it is convenient to convert each \pda\ in a normal form in which in a single move the topmost symbol of the pushdown store cannot be replaced. 
A such a form could be easily obtained by simulating each move that replaces the topmost symbol~$A$ with a string~$\alpha=\alpha'B$, $B\neq A$, by a move
that pops~$A$, followed by a move that pushes the string~$\alpha'$. However, such a construction
does not preserve the number of turns.
In the following technical lemma we present an alternative simulation that preserves the number of turns:

\begin{lemma}
	\label{lemma:normalForm}
	For each \pda~${\cal M}=(Q,\Sigma,\Gamma,\delta,q_0,Z_0)$ there exists an equivalent \pda~${\cal M}'=(Q',\Sigma,\Gamma',\delta',q'_0,\bot)$ such that 
	if~$(p,\alpha B)\in\delta'(q,\sigma,A)$, with~$p,q\in Q'$, $\alpha\in\Gamma'^*$, $\sigma\in\Sigma\cup\{\varepsilon\}$, $A,B\in\Gamma'$, then~$A=B$.
	Furthermore, for each~$w\in\Sigma^*$ and~$k\in\IN$, ${\cal M}$ has an accepting computation on~$w$ making exactly~$k$ turns if and only if~${\cal M}'$
	has an accepting computation on~$w$ making exactly~$k$ turns.
\end{lemma}
\begin{proof}
The \pda~${\cal M}'$ makes a direct simulation of the moves of the given \pda~${\cal M}$, by keeping in its finite control the symbol which is on the top of the pushdown store of~${\cal M}$.
Hence, the pushdown contents of~${\cal M}$ is represented by the symbol in the finite control, plus the symbols in the store  of~${\cal M}'$.
To do that, when a transition of~${\cal M}$ replaces the topmost symbol~$A$ by a nonempty string~$\alpha$, in the \pda~${\cal M}'$ the symbol~$A$, which is kept in the control, is replaced by the first symbol of~$\alpha$, while the other symbols of~$\alpha$ are pushed on the stack.
When a transition of~${\cal M}$ pops the topmost symbol off the stack, the corresponding transition of~${\cal M}'$ also pops off the topmost symbol, which is the `new' top for~${\cal M}$, and saves it in the  control.
For technical reasons, it is useful to introduce a new pushdown symbol~$\bot$, which is used to mark the bottom of the store of~${\cal M}'$. This symbol is removed only when the store of~${\cal M}$ becomes empty, so ending the computation because the next transition is not defined.

Formally, we take~$\Gamma'=\Gamma\cup\{\bot\}$, where~$\bot\notin\Gamma$ is the only symbol on the pushdown store at the beginning of the computation, and~$Q'=Q\times\Gamma'$.
For brevity, states in~$Q'$ are represented as~$[qA]$, where~$q\in Q$ and~$A\in\Gamma'$.
The initial state~$q'_0$ is~$[q_0Z_0]$.

The transition function~$\delta':Q'\times(\Sigma\cup\{\varepsilon\})\times\Gamma'\rightarrow 2^{Q'\times\Gamma'^*}$ is defined by considering the transitions listed below,
where~$\sigma\in\Sigma\cup\{\varepsilon\}$, $q,p\in Q$, $A\in\Gamma$, $\alpha\in\Gamma^*$:
\begin{enumerate}
	\item[(a)]For each transition~$(p,\varepsilon)\in\delta(q,\sigma,A)$ and for each~$C\in\Gamma\cup\{\bot\}$, $\delta'([qA],\sigma,C)$ contains~$([pC],\varepsilon)$;
	\item[(b)]For each transition~$(p,B_1B_2\cdots B_k)\in\delta(q,\sigma,A)$, with~$B_1,B_2,\ldots,B_k\in\Gamma$, $k>0$,  and for each~$C\in\Gamma\cup\{\bot\}$, $\delta'([qA],\sigma,C)$ contains~$([pB_1],B_2\cdots B_kC)$.
\end{enumerate}
We now use standard notations for configurations and moves of~\pdas\ as in\cite{HU79}: a configuration is a triple~$(q,x,\alpha)$ where $q$ is the current state, $x$ is the part of the input which has not been yet read, i.e., that starts from the current input head position, and~$\alpha$ is the string on the pushdown store, where the topmost symbol is at the beginning, and~$(q,x,\alpha)\Nreaches{}(q',x',\alpha')$ denotes that from the configuration~$(q,x,\alpha)$ in one step a \pda~$\cal N$ can reach configuration~$(q',x',\alpha')$.
The reflexive and transitive closure of the relation~$\Nreaches{}$ is denoted as~$\Nreaches{*}$.

From~(a),
we can observe that for~$q,p\in Q$, $\sigma\in\Sigma\cup\{\varepsilon\}$, $x\in\Sigma^*$,
$A,C_1,C_2,\ldots,C_s\in\Gamma$, if~$s>0$ then:
\begin{eqnarray*}
&&(q,\sigma x, AC_1C_2\cdots C_s)\Mreaches{}(p,x,C_1C_2\cdots C_s)
\mbox{ if and only if }\\
&&([qA],\sigma x, C_1C_2\cdots C_s\bot)\Preaches{}([pC_1],x,C_2\cdots C_s\bot)\,,
\end{eqnarray*}
while, when~$s=0$:
\begin{eqnarray*}
(q,\sigma x, A)\Mreaches{}(p,x,\varepsilon)
\mbox{ if and only if }
([qA],\sigma x, \bot)\Preaches{}([p\bot],x,\varepsilon)\,.
\end{eqnarray*}
From~(b), we obtain that for~$q,p\in Q$, $\sigma\in\Sigma\cup\{\varepsilon\}$, $x\in\Sigma^*$,
$A,B_1,B_2,\ldots,B_k,C_1,C_2,\ldots,C_s\in\Gamma$ with~$k>0,s\geq 0$:
\begin{eqnarray*}
&&(q,\sigma x, AC_1C_2\cdots C_s)\Mreaches{}(p,x,B_1B_2\cdots B_kC_1C_2\cdots C_s)
\mbox{ if and only if }\\
&&([qA],\sigma x, C_1C_2\cdots C_s\bot)\Preaches{}([pB_1],x,B_2\cdots B_kC_1C_2\cdots C_s\bot)\,.
\end{eqnarray*}
Hence, to each sequence of consecutive configurations of~$\cal M$:
\begin{eqnarray*}
(q_1,w_1,A_1\alpha_1)\Mreaches{}(q_2,w_2,A_2\alpha_2)\Mreaches{}\cdots\Mreaches{}(q_k,w_k,A_k\alpha_k)
\end{eqnarray*}
corresponds the sequence of configurations of~${\cal M}'$:
\begin{eqnarray*}
([q_1A_1],w_1,\alpha_1\bot)\Mreaches{}([q_2A_2],w_2,\alpha_2\bot)\Mreaches{}\cdots\Mreaches{}([q_kA_k],w_k,A_k\bot)
\end{eqnarray*}
and vice versa. Notice that~$|A_i\alpha_i|=|\alpha_i\bot|$ for~$i=1,\ldots,k$. This implies that the two sequences contain
exactly the same number of turns. Furthermore, if the last configuration in the first sequence
is replaced by~$(q_k,w_k,\varepsilon)$, i.e., the last move pops~$A_{k-1}$ off the pushdown and~$\alpha_{k-1}=\varepsilon$, then the last configuration in the corresponing sequence becomes~$([q_k\bot],w_k,\varepsilon)$.

This allows to conclude that for each~$w\in\Sigma^*$,
$(q_0,w,Z_0)\Mreaches{*}(q,\varepsilon,\varepsilon)$ if and only if
$([q_0Z_0],w,\bot)\Preaches{*}([q\bot],\varepsilon,\varepsilon)$, i.e., the languages accepted by~$\cal M$ and~${\cal M}'$ coincide.
From the previous discussion we also conclude that~$\cal M$ has a computation accepting a string~$w$ in exactly~$t$ turns if and only if~${\cal M}'$ has a computation accepting~$w$ in exactly~$t$ turns.
\qed
\end{proof}

We are now able to prove a lower bound for the number of turns necessary to accept some
generalizations of languages~$\eq^k$, $k>0$.

\begin{lemma}\label{lemma:lowerBoundEQ}
	Let~$\Sigma=\Delta\cup\{a,b\}$ be an alphabet, with $a,b\notin\Delta$, and let~$\ma{M}$ be a \pda\ 
	accepting a language~$L\subseteq\Sigma^*$.
	There exists an integer~$N>0$ such that for each integer~$k>0$ and strings~$y_0,y_1,\ldots,y_k\in\Delta^*$,
	if the language~$L'=y_0\{a,b\}^*y_1\{a,b\}^*\cdots y_{k-1}\{a,b\}^*y_k\cap L$ is a subset 
	of~$y_0\eq^*y_1\eq^*\cdots y_{k-1}\eq^*y_k$ then, for~$n_1,n_2,\ldots,n_k\geq N'=N\cdot(|y_0y_1\cdots y_k|+1)$,
	each accepting computation
	of~$\ma{M}$ on input~$y_0a^{n_1}b^{n_1}y_1a^{n_2}b^{n_2}\cdots y_{k-1}a^{n_k}b^{n_k} y_k$ contains at least~$k$ turns.\footnote{%
Notice that the statement does not imply the existence of an accepting computation on~$y_0a^{n_1}b^{n_1}y_1a^{n_2}b^{n_2}\cdots y_{k-1}a^{n_k}b^{n_k} y_k$.
	}\,\footnote{%
	We use this lemma in the proofs of Theorems~\ref{th:nonRec}, \ref{th:sq}, \ref{th:loglog-lb}, and~\ref{th:lstar-lower}.
	Actually, we always obtain a language~$L'=y_0\{a,b\}^*y_1\{a,b\}^*\cdots y_{k-1}\{a,b\}^*y_k\cap L$
		which is a subset of~$y_0\eq y_1\eq\cdots y_{k-1}\eq y_k$. However, we prefer to leave the statement in this more general form that derives from the proof.
	}
\end{lemma}
\begin{proof}
	Without loss of generality, according to Lemma~\ref{lemma:normalForm}, we can suppose that for each transition of~$\cal M$, 
	of the form~$(q,\alpha)\in\delta(q,\sigma,A)$, with~$\alpha\neq\varepsilon$,
	the last symbol of~$\alpha$ is~$A$.
	
	From~$\ma{M}$, by applying a standard construction for the intersection with regular languages, we can obtain
	a \pda~$\ma{M}'$ accepting~$L'$, in such a way that there is a one--to--one correspondence between
	accepting computations of~$\ma{M}$ and of~$\ma{M}'$ on inputs belonging to~$L'$.
	An accepting computation in~$\ma{M}$ and the corresponding computation in~$\ma{M}'$ perform exactly
	the same number of turns.	
	
	We give a few details on~$\ma{M}'$ that will be used in the proof.
	We take a finite automaton~$A$ accepting the language~$R=y_0\{a,b\}^*y_1\{a,b\}^*\cdots y_{k-1}\{a,b\}^*y_k$.
	The transition graph of~$A$ consists of a path of~$|y_0y_1\cdots y_k|+1$ states,
	from the initial state to the final state. The concatenations of the letter labeling the arcs in this path is~$y_0y_1\cdots y_k$.
	The only loops in the graph are self loops, labeled with letters~$a$ and~$b$, that are located on the states corresponding
	to the end of the factor~$y_{i-1}$ and the beginning of~$y_i$, for~$i=1,\ldots,k$.
	Hence, the states of~$\ma{M}'$ are pairs~$(q,r)$ where~$q$ is a state of~$\ma{M}$ and~$r$ is a state~$A$.
	Let us take~$N=(\#Q)^2\cdot\#\Gamma$, where~$Q$ and~$\Gamma$ are the set of states and the working alphabet
	of~$\ma{M}$, respectively.
	
	Suppose that~$x=y_0a^{n_1}b^{n_1}y_1a^{n_2}b^{n_2}\cdots y_{k-1}a^{n_k}b^{n_k} y_k$, with~$n_1,n_2,\ldots,n_k\geq N'=N\cdot(|y_0y_1\cdots y_k|+1)$, 
	is accepted by~$\ma{M}$ using less
	than~$k$ turns. Then there also exists a computation~$\cal C$ of~$\ma{M}'$ accepting~$x$ using less than~$k$ turns.
	We will show that this leads to a contradiction.

	Since the number of turns in~$\cal C$ is less than~$k$, there must exists an index~$i$, $1\leq i\leq k$, such that
	no turns are performed while consuming the factor~$a^{n_i}b^{n_i}$.
	Hence, the part~${\cal C}'$ of~$\cal C$ consuming~$a^{n_i}b^{n_i}$ cannot contain any pop, after a push has been executed.
	This leads to consider two cases.
	
	If a push is executed in~${\cal C}'$ \emph{before} the input head reaches the factor~$b^{n_i}$
	then, while reading such a factor, the height of the pushdown store cannot decrease.
	For~$j=0,\ldots,n_i$, let us consider the integer~$h_j$ and the quintuple~$[q_jr_jA_jp_js_j]$ where, for~$j=0,\ldots,n_i$, 
	$h_j$ is the height of the pushdown store immediately after reading the prefix~$a^{n_i}b^j$ in~${\cal C}'$,
	$(q_j,r_j)$ is the state of the control of~$\ma{M}'$, $A_j$ is the symbol on the top of the store, while~$(p_j,s_j)$ is the
	state in the first configuration of~$\cal C$ after all the factor~$b^{n_i}$ has been consumed, 
	in which the height of the pushdown store is~$h_j$. (Roughly, while reading~$b^{n_i}$, the height of the pushdown store could
	increase over~$h_j$, hence at some point in the next part of computation, after some pop operations the height
	of the pushdown store should decrease again to~$h_j$.)
	We observe that, since in the above defined quintuples the first two components correspond to a part of the computation on the same factor~$b^{n_i}$,
	the component~$r_j$ corresponding to the state of~$A$ after reading~$b^j$ does not change. Hence,
	the number of different quintuples that we are considering for the given~$b^{n_i}$ is bounded by the number of triples 
	belonging to~$(\#Q)^2\cdot\#\Gamma$, multiplied by the number of possible values of the component~$s_j$, namely~$|y_0y_1\cdots y_k|+1$
	Since~$n_i\geq N'$, the number of quintuple we are considering exceeds such an amount. 
	So there are two indices~$j,\ell$, with~$0\leq j < \ell \leq n_i$ such that~$[q_jr_jA_jp_js_j]=[q_{\ell}r_{\ell}A_{\ell}p_{\ell}s_{\ell}]$.
	Furthermore~$h_j\leq h_{\ell}$. 
		
	If~$h_j<h_{\ell}$ then this defines a `vertical loop'~\cite{PP23}, which consumes an input factor~$b^{\ell-j}$ in the increasing phase
	and an input factor~$z$ in the decreasing phase. This loop can be used to pump the computation~$\cal C$ in order
	to obtain a computation of~$\ma{M}'$ accepting a string~$x'\notin L'$, which is a contradiction.
	
	Otherwise~$h_j=h_{\ell}$, thus implying that the non-empty factor~$b^{\ell-j}$
	is consumed between two repetitions of the state~$q_j$, without changing the  pushdown store in between.
	This also allow to pump the computation~$\cal C$ by doubling the part between these two repetitions,
	obtaining again an accepting computation on a string~$x'\notin L'$, so a contradiction.
	
	If no push is executed in~${\cal C}'$, \emph{before} the input heads reaches the factor~$b^{n_i}$, then the 
	stack height can only decrease or remain stationary while reading~$a^{n_i}$. With an argument symmetrical to the previous one we obtain a contradiction. 
\qed
\end{proof}

\section{Undecidability and Non-Recursive Trade-Offs}
\label{sec:undec}

Ginsburg and Spanier proved that it is decidable whether a \pda\ is finite turn or whether it is~$k$-turn
for any given integer~$k\geq 0$~\cite{GS66}.
We underline that all the accepting computations are considered to count the turns.
So, for each accepted input, the bound must be satisfied by the most expensive accepting computation.

Here, we consider the number of turns which are \emph{sufficient} to accept a given input. Hence, the bound
must be satisfied, for each accepted input, by the accepting computation using the lowest number of turns.
With this measure, we obtain completely different results. In particular, 
all the above mentioned questions become undecidable, except acceptance in~$0$ turns.
We will also prove a non-recursive trade-off result, by showing that in the case of acceptance in a number of
turns bounded by a value~$k$, constant with respect to the length of the input, such a value cannot be bounded
by any recursive function in the size of the description of the machine.
This will allow to prove that the trade-offs between \ocas\ and \pdas\ accepting in a finite number of turns and
finite-turn \pdas\ are not recursive.

\smallskip

Let us start by presenting the undecidability results, that are obtained by adapting a well-known
technique introduced by Hartmanis and which is based on encodings of single-tape Turing machine computations~\cite{Har67}.
	
Configurations of a single-tape Turing machine~$\ma{T}$ with state set~$Q$ and tape alphabet~$\Gamma$
are denoted, in a standard way, as strings in~$\Gamma^*Q\Gamma^*$.
A computation is encoded as a string~$\alpha_1\$\alpha_2\$\cdots\alpha_m\$$,
where the blocks~$\alpha_i$, $i = 1, \ldots, m$,
are encodings of configurations of~$\ma{T}$
and~$\$\notin Q\cup\Gamma$ is a delimiter.
Hence,
the (encoding of a) \emph{valid computation} of~$\ma{T}$ on input~$w$
is a string~$\alpha_1\$\alpha_2\$\cdots\alpha_m\$$,
for some integer~$m\geq 1$, such that:
\begin{enumerate}
	\item\label{en:conf}$\alpha_i\in\Gamma^*Q\Gamma^*$,
		i.e., $\alpha_i$ encodes a configuration of~$\ma{T}$, $i=1,\ldots,m$;
	\item\label{en:init}$\alpha_1$ encodes the initial configuration on input~$w$
	 by the string~$q_Iw$, where~$q_I$ is the initial state of~$\ma{T}$;
	\item\label{en:halt}$\alpha_m$ encodes a halting configuration of~$\ma{T}$,
		namely a configuration from which no move is possible;
	\item\label{en:step}$\alpha_{j+1}$ is reachable in one step from~$\alpha_j$,
		$j=1,\ldots,m-1$.
\end{enumerate}
Let us denote by~$\Valid{\ma{T}\!,w}$ the set of valid computations of~$\ma{T}$ on input~$w$
and by~$\Valid{\ma{T}\!,w}^c$ its complement.

\begin{lemma}
	\label{lem:oca invalid}
	Let~$\ma{T}$ be a  Turing machine and~$w$ be a fixed string.\linebreak[4]
	The set~$\Valid{\ma{T}\!,w}^c$ is accepted by a $1$-turn \oca.
\end{lemma}
\begin{proof}
	The language~$\Valid{\ma{T}\!,w}^c$ is accepted by a \oca~$\ma{M}_{\ma{T}\!,w}$ using the following strategy.
	Given~${\cal D}=\beta_1\$\beta_2\$\cdots\$\beta_m\$$,
	with~$\beta_j\in(Q\cup\Gamma)^*$, $j=1,\ldots,m$,
	in order to decide whether~${\cal D}\in~\Valid{\ma{T}\!,w}^c$,
	$\ma{M}_{\ma{T}\!,w}$ guesses which one among Conditions~\ref{en:conf}, \ref{en:init}, \ref{en:halt}, and  \ref{en:step} is not satisfied.
	For the first three conditions,
	the verification of the guess is done by only using the finite control.
	For the last condition,
	$\ma{M}_{\ma{T}\!,w}$ can guess a block~$j$, $1\leq j< m$, and an ``un-matching'' position between blocks~$j$ and~$j+1$.
	The verification uses the counter to locate the corresponding positions in the two blocks.
\qed
\end{proof}

We are now ready to prove the main result of this section:

\begin{theorem}\label{th:oc-undec}
	The following properties are undecidable:
	\begin{itemize}
	\item whether a \oca\ accepts in a finite number of turns,
	\item whether a \oca\ accepts in~$k$ turns, for a given integer~$k>0$.
	\end{itemize}
\end{theorem}
\begin{proof}
	We give a reduction from the blank-tape halting problem, namely the problem of deciding whether a Turing machine
	halts on the empty tape.
	Given a Turing machine~$\ma{T}$ with state set~$Q$ and working alphabet~$\Gamma$, s.t.~$\$\notin Q\cup\Gamma$,
	let us consider the language
	\[L=\{x\$y\mid x\in\{a,b\}^*, y\in(Q\cup\Gamma\cup\{\$\})^*,\mbox{ and }(x\in\eq^*\mbox{ or } y\notin\Valid{\ma{T}\!,\varepsilon})\}\,.\]
	We build a \oca~$\ma{M}$ that given an input~$z$ of the form~$x\$y$, with~$x\in\{a,b\}^*$, decides if~$z\in L$
	by nondeterministically choosing between verifying that~$x\in\eq^*$ (this is done by simulating the \oca~$\meqs$),
	or verifying that~$y\notin\Valid{\ma{T}\!,\varepsilon}$, which is done by simulating a $1$-turn \oca\ accepting $\Valid{\ma{T}\!,\varepsilon}^c$.
	
	If~$\ma{T}$ halts on input~$\varepsilon$, then fix~$\hat{y}\in\Valid{\ma{T}\!,\varepsilon}$. For each~$k>0$, 
	the only accepting computation on input~$(ab)^k\$\hat{y}$ uses~$k$ turns.
	So the number of turns used by~$\ma{M}$ to accept cannot be bounded by any constant.
	
	However, if~$\ma{T}$ does not halt on input~$\varepsilon$, then $\Valid{\ma{T}\!,\varepsilon}$ is empty and, hence,
	each input~$z=x\$y$, with~$x\in\{a,b\}^*$, is accepted by~$\ma{M}$ in a computation that verifies 
	that~$y\in\Valid{\ma{T}\!,\varepsilon}^c$ using at most one turn.
	
	Hence, we can conclude that~$\ma{M}$ accepts in a finite number of turns if and only if~$\ma{T}$ does not halt on input~$\varepsilon$,
	which is undecidable.
	So, it cannot be decided if a \oca\ accepts in a finite number of turns.
	
	The second statement follows from the fact that when the above-described machine~$\ma{M}$ accepts in a finite number of turns, 
	then it accepts in one turn and, so, in~$k$ turns for each~$k>0$.
\qed
\end{proof}

We now prove that restricting to~$k=0$, the second property considered in Theorem~\ref{th:oc-undec} becomes decidable even in the case
of \pdas.
We point out that acceptance in~$0$ turns implies that the pushdown store is useless, hence the language accepted by the machine is regular.
However, the property we are investigating is different from deciding whether the accepted language is regular, which is 
undecidable~\cite{BPS61}. Indeed, there are \pdas\ that, even making use of turns, accept regular languages.

\begin{theorem}
\label{th:0turns}
	It is decidable whether a \pda~$\ma{M}$ accepts in~$0$ turns.
\end{theorem}
\begin{proof}
Let us suppose that the start symbol of the pushdown store of~$\ma{M}$ is~$\bot$. We can suppose that~$\bot$
is removed only by a $\varepsilon$-transitions, so leading to acceptance exactly when all the
input has been scanned (this can by obtained, e.g., using the construction in Lemma~\ref{lemma:normalForm}).

		By removing from~$\ma{M}$ all the transitions making push or pop operations and choosing as final all the states of~$\ma{M}$ from which~$\bot$ is popped off the stack, we obtain a finite automaton~$\ma{A}$ accepting
		all the strings that are accepted by~$\ma{M}$ without using any turn. 
		Then~$\ma{M}$ accepts in $0$ turns if and only if~$L(\ma{A})=L(\ma{M})$ that, being~$L(\ma{A})\subseteq L(\ma{M})$,
		is equivalent to~$L(\ma{M})\cap(L(\ma{A}))^c=\emptyset$.
		Since context-free languages are effectively closed by intersection with regular languages and
		emptiness is decidable for context-free languages, we conclude that the property is decidable.
\qed
\end{proof}

%

The last part of this section is devoted to investigate size trade-offs.

%
\begin{theorem}
\label{th:nonRec}
	There exists a function~$T:\IN\rightarrow\IN$ that cannot be bounded by any recursive function, such that
	for infinitely many integers~$s$ there exists a context-free
	language~$L_s$ having a description of size~$s$ such that:
	\begin{itemize}
		\item There exists a \oca\ accepting~$L_s$ in~$T(s)$ turns.\footnote{We point out that the parameter~$s$ is the
		size of the description of the language~$L_s$. Hence, $T(s)$ is constant with
		respect to the input length. This means that for each string in~$L_s$ the \oca\
		   has an accepting computation making at most~$T(s)$ turns.}
		\item Each \pda\ accepting~$L_s$ uses at least~$T(s)$ turns to accept strings of length~$n$, for infinitely many integers~$n$.
	\end{itemize}
\end{theorem}
\begin{proof}
The argument is derived from~\cite[Prop.~7]{MF71}.
Let us start by introducing the following modification of the complement of~$\Valid{\ma{T}\!,w}$, in which 
each occurrence of the delimiter~$\$$ is replaced by a string in~$\{a,b\}^+$. 
For technical reasons and without loss of generality, let us suppose~$a,b\notin Q\cup\Gamma\cup\{\$\}$ where, as usual, $Q$ is the state set, $\Gamma$ is the working alphabet, and~$\$\notin Q\cup\Gamma$:
\begin{eqnarray*}
\pNotValid{\ma{T}\!,w} &=&\{\alpha_1z_1\alpha_2z_2\cdots \alpha_mz_m  \mid m>0,\\ 
			& &~~~~\alpha_i\in(\Gamma\cup Q)^+, z_i\in\{a,b\}^+,~i=1,\ldots,m, \mbox{ and:}\\
			& &~~~~\mbox{\rm(a)~} \alpha_1\$\alpha_2\$\cdots\alpha_m\$\in\Valid{\ma{T}\!,w}^c, \mbox{ or }\\
			& &~~~~\mbox{\rm(b)~} \forall i, 1\leq i\leq m,\,  z_i\in\eq\,\}\,.
\end{eqnarray*}
	By adapting the arguments used to prove Theorem~\ref{th:oc-undec},
	we define a \oca~$\ma{M}$ that accepts 
	$\pNotValid{\ma{T},\varepsilon}$ with the following strategy.
	At the beginning of the computation, on input~$x=\alpha_1z_1\alpha_2z_2\cdots\alpha_mz_m$,~$\ma{M}$
	nondeterministically decides to verify one of the two conditions~(a) and~(b). 
	Condition~(a) can be tested, as described in the proof of Theorem~\ref{th:oc-undec}, using at most one turn.
	To test condition~(b), the input is scanned while simulating the \oca~\meq\ on each block~$z_i$. This uses exactly~$m$ turns.
	We point out that the size of the description of~$\ma{M}$ is polynomial in the number of states
	of the given machine~$\ma{T}$.
	
	Let us consider any halting single-tape deterministic Turing machine~$\ma{T}_s$ with a description of size~$s$.
	Let~${\cal C}_s$ be the encoding of the computation of~$\ma{T}_s$ on~$\varepsilon$. It is well known that
	the length of~${\cal C}_s$ cannot be bounded by any recursive function in the number of states of~$\ma{T}_s$ and, so, in~$s$.
	
	Let us denote by~$\ma{M}_s$ the \oca\ accepting~$\pNotValid{\ma{T}_s,\varepsilon}$ obtained according to the previous construction.
	
	Suppose that~$\ma{T}_s$ halts on~$\varepsilon$. Being~$\ma{T}_s$ deterministic, 
	there is only one computation~${\cal C}_s$ on~$\varepsilon$.
	This computation is described by a string~$\alpha_1\$\alpha_2\$\cdots\alpha_m\$\in\Valid{\ma{T}_s,\varepsilon}$.
	This implies that each string~$x=\alpha_1z_1\alpha_2z_2\cdots\alpha_mz_m$, with~$z_i\in\eq$, for $i=1,\ldots,m$, satisfies 
	\emph{only} condition~(b). Hence, it requires~$m$ turns to be accepted, regardless its length.
	We take~$T(s)=m$.
	Since, the number~$m$ of configurations occurring in~${\cal C}_s$ cannot be bounded by any recursive function in~$s$,
	it follows that the function~$T(s)$ cannot be bounded by any recursive function in~$s$.
	
	To prove the lower bound on the number of turns necessary to accept~$L_s$ by any \pda, we consider again the encoding~$\alpha_1\$\alpha_2\$\cdots\alpha_m\$$ of
	 the computation~${\cal C}_s$ on~$\varepsilon$.
	 The intersection of~$\alpha_1\{a,b\}^*\alpha_2\{a,b\}^*\cdots\alpha_m\{a,b\}^*$ with~$\pNotValid{\ma{T}_s,\varepsilon}$ 
	 coincides with~$\alpha_1\eq\alpha_2\eq\cdots\eq\alpha_m\eq$. Hence, according to Lemma~\ref{lemma:lowerBoundEQ}, 
	 for each \pda\ accepting~$\pNotValid{\ma{T}_s,\varepsilon}$, $m$ turns are necessary
	 to recognize any input of the form~$\alpha_1z\alpha_2z\cdots\alpha_mz$, with~$z=a^Kb^K$, for each sufficiently large~$K$.
\qed
\end{proof}

In contrast to Theorem~\ref{th:nonRec}, by considering all accepting computations, we obtain a recursive bound:

\begin{theorem}
\label{th:rec}
	There exists an exponential function~$f$ such that if~$\ma{M}$ is a~$k$-turn \pda\ then~$k\leq f(s)$, where~$s$ is the size of~$\ma{M}$.
\end{theorem}
\begin{proof}
	We adapt and refine the construction used to prove that it is decidable whether a \pda\ is finite turn~\cite[Thm.~5.1]{GS66}.
	Given a \pda~$\ma{M}$, we can modify it in order to keep track in its finite control of the current `phase' of the pushdown store,
	namely if it is increasing or decreasing.
	For this purpose, the set of state~$Q$ is extended with a component in~$\{\uparrow,\downarrow\}$, which is set to~$\uparrow$ in
	the initial state.
	While simulating a pop move, if the second component was~$\uparrow$, then it is switched to~$\downarrow$. Conversely,
	for a push move, if the second component was~$\downarrow$ then it is changed into~$\uparrow$. In the other cases the second component is unchanged.
	The resulting \pda\ is equivalent to the original one.
	
	Now, we further modify the machine, by replacing the input symbols in the transitions in order to obtain a \pda\ over the unary alphabet~$\{a\}$
	that counts in its input the turns. 
	Each transition corresponding to a turn, namely in which the second component changes from~$\uparrow$ to~$\downarrow$, reads the symbol~$a$
	from the input, while all the other transitions become~$\varepsilon$-transitions.
	In this way, the resulting \pda~$\ma{M}_a$ accepts the following language:
	\[
	L_a=\{a^t\mid~\ma{M}~\mbox{as an accepting computation containing~$t$ turns}\}\,.\footnote{%
		Actually, since the acceptance is by empty stack, an accepting computation with~$0$ turns, in the modified version of~$\ma{M}$,
 		consists of a sequence of configurations in which the second component of the state is~$\uparrow$, ending in a configuration, after
		a pop of bottom symbol~$Z_0$, in which the second component is~$\downarrow$.
		Hence, the string~$a^1$ in~$L_a$ can represent computations making~$0$ or~$1$ turns.}
	\]
	According to~\cite{PSW02}, from~$\ma{M}_a$ we can get a nondeterministic finite automaton~$A$ accepting~$L_a$ with a number of states bounded
	by~$f(s)$, where~$f$ is an exponential function.
	
	If~$\ma{M}$ is~$k$-turn, then~$L_a$ is finite and contains only strings of length bounded by~$k$. Hence the automaton~$A$ should contain at least~$k+1$ states.
	This allows to conclude that~$k<f(s)$.
\qed
\end{proof}

As a consequence of Theorems~\ref{th:nonRec} and~\ref{th:rec} we obtain:

\begin{corollary}
	\label{cor:tradeOff}
	The trade-offs between \ocas\ and \pdas\ accepting in a constant number of turns and finite-turn \pdas\ are not recursive.
\end{corollary}

\section{A Sublinear Number of Turns}
\label{sec:sublinear}

In this section, we present a language~$\Lsq$ which is accepted in a sublinear number of turns.
This language is a `padded' version of a language presented in~\cite{Gab84}, whose strings of length~$n$ can be accepted by a \oca\ using a counter
bounded by~$O(\sqrt{n})$. 
We will prove that~$\Lsq$ is accepted by a \oca\ which uses~$O(\sqrt{n})$ turns. Furthermore, this bound cannot be significantly improved, 
even using a \pda.

The language~$\Lsq$ is over the alphabet~$\{0,a,b\}$. Roughly speaking, a string belongs to it if the sequence of its maximal factors in~$0^+$
does not coincide with the sequence of unary representations of all integers from~$1$ to some~$m\geq 1$, or all its maximal factors in~$\{a,b\}^+$
belong to the language~$\eq$:
\begin{eqnarray*}
\Lsq =\{x_1z_1x_2z_2\cdots x_mz_m&\mid&m>0, x_i\in\{0\}^+, z_i\in\{a,b\}^+, i=1,\ldots,m, \mbox{ and:}\\
			& &\mbox{\rm~~(a)~} \exists i, 1\leq i\leq m, |x_i|\neq i, \mbox{ or }\\
			& &\mbox{\rm~~(b)~} \forall i, 1\leq i\leq m,\,  z_i\in\eq\,\}\,.
\end{eqnarray*}

We now prove that~$\Lsq$ is accepted by a \oca\ in a number of turns of the order of~$\sqrt n$.
Furthermore, we prove a lower of the order of~$\sqrt[3]{n}$ which holds also in the case of \pdas.

\begin{theorem}\label{th:sq}
	The language~$\Lsq$ is accepted by a \oca\ in~$O(\sqrt n)$ turns.
	Furthermore, each \pda\ cannot accept~$\Lsq$ using~$o(\sqrt[3]{n})$ turns.
\end{theorem}
\begin{proof}
Given a string~$w=x_1z_1x_2z_2\cdots x_mz_m$, a \oca~$\ma{A}$ can decide if~$w\in \Lsq$
by nondeterministically selecting one between the conditions~(a) and~(b), and verifying it.
(We are supposing that~$w$ is in the appropriate format. This can be tested using the finite control while scanning~$w$.)

We observe that condition~(a) is true if and only if either~$x_1\neq 1$ or there exists~$i$,
$1\leq i<m$, such that~$|x_{i+1}|\neq|x_i|+1$. 
Hence, it can be tested by checking if~$x_1\neq 1$ or nondeterministically choosing a block~$x_i$, saving~$|x_i|$ in the counter and, after performing one turn,
using the counter to verify the length of the next block, if any.
If the verification is successful then the machine accepts.

Condition~(b) can verified by simulating on each block~$z_i$ the \oca~$\meq$ accepting~$\eq$. This uses one turn for each~$z_i$, hence~$m$ turns.

Since to accept strings satisfying~(a) one turn is enough,
to obtain the total number of turns we estimate the number~$m$ of blocks with respect to the length~$n$ of any string~$w$ 
for which \emph{only} condition~(b) is satisfied, 
namely~$w=0^1z_10^2z_2\cdots 0^mz_m$, with~$m>0$ and~$z_1,z_2,\ldots,z_m\in\eq$.
This gives~$n=|w|>\sum_{i=1}^mi ={m(m+1)}/{2}$.
Hence, the number~$m$ of turns used to accept~$w$ is~$O(\sqrt n)$.

\smallskip

To prove the lower bound, let us consider a \pda~$\ma{M}$ accepting~$\Lsq$.
	Let~$N$ be the constant of Lemma~\ref{lemma:lowerBoundEQ} for~$\ma{M}$. 
	Given an integer~$m>0$, according to the definition of~$\Lsq$,
	the intersection of~$0^1\{a,b\}^*0^2\{a,b\}^*\cdots 0^m\{a,b\}^*$ with~$L$ 
	is~$0^1\eq 0^2\eq\cdots 0^m\eq $.
	Let us take~$N'=N\cdot({m(m+1)}/{2}+1)$.
	By Lemma~\ref{lemma:lowerBoundEQ}, to be accepted by~$\ma{M}$, the string~$w=0^1a^{N'}b^{N'}0^2a^{N'}b^{N'}\cdots 0^ma^{N'}b^{N'}$
	requires~$m$ turns.
	Furthermore~$n = |w| = \frac{m(m+1)}{2}+2N'm=\Theta(m^3)$.
	Hence, for each~$m$ there is an input length~$n$ for which~$m$ turns are required, where~$m = \Theta(\sqrt[3]{n})$. 
	This allows to conclude that~$\ma{M}$ cannot accept using~$o(\sqrt[3]n)$ turns.
\qed
\end{proof}

\section{An Infinite Hierarchy}
\label{sec:hie}

In Section~\ref{sec:sublinear}, we presented a language~$\Lsq$ accepted in a sublinear number of turns, 
not growing less than~$\sqrt[3]{n}$.
We now show that there are languages for which the number of turns, necessary and sufficient, can be dramatically reduced.
In particular, we will prove that for each integer~$k$ there is a language~$\Lk{k}$ accepted in a number of turns of the order of~$\logk{k}n$
which is also necessary. This gives an infinite hierarchy with respect to the number of turns.

\medskip

To obtain such result, let us start by introducing some preliminary tools.
Let us denote by~$\bin{x}$ the binary representation of any integer~$x>0$. 
It is well known that~$|\bin{x}|=\lfloor\Log x\rfloor+1$.

Given a string~$y$ of the form~$y=x_1\$x_2\$\cdots x_m\$$, with~$x_i\in\{0,1\}^+$, $i=1,\ldots,m$, let us denote by~$\last{y}$ the rightmost binary block in~$y$,
namely the block~$x_m$.

We consider the following language over the alphabet~$\{0,1,\$\}$:
\[
	\ListBin=\{ \bin{1}\$\bin{2}\$\cdots\$\bin{m}\$ \mid m\geq 1\}\,.
\]
Strings in~$\ListBin$ are lists of binary encoded integers from~$1$ to some~$m\geq 1$, where each encoding is ended by the
delimiter~$\$$.

Using a pumping argument, it can be proved that~$\ListBin$ is not a context-free language. However, its complement is context free.
More precisely, it is accepted by a $1$-turn \oca~\cite{LSH65,Gab84}.

With abuse of notation, for~$y\in\ListBin$, by~$\last{y}$ we also denote the integer represented by the rightmost binary block in~$y$.
\begin{lemma}
\label{lemma:Lbin}
	If~$y=x_1\$x_2\$\cdots\$x_m\$\in\ListBin$, with~$x_i\in\{0,1\}^*$, $i=1,\ldots,m$, then:
	\begin{itemize}
	\item $|\last{y}|\leq\Log|y|$,
	\item $\occurr{y}{\$}=m=\last{y}\geq 2^{|\last{y}|-1}$,
	\item $|y|\leq 2m+m\Log m$.
	\end{itemize}
\end{lemma}
\begin{proof}
	Since~$y\in\ListBin$, $x_i=\bin{i}$ for~$i=1,\ldots,m$.
	Furthermore, because the delimiters, $m\leq|y|/2$.
	Hence:
	\[|\last{y}|=|x_m|=|\bin{m}|\leq 1+\Log m\leq 1+\Log\frac{|y|}{2}=\Log|y|\,.\]	
	The equalities in the second statement are  trivial.
	Furthermore,~$|\last{y}|=|\bin{m}|=\lfloor\Log m\rfloor+1$.
	Hence, $2^{|\last{y}|-1}=2^{\lfloor\Log m\rfloor}\leq m$.
	
	For the third statement, we observe that~$\sum_{i=1}^m\lfloor\Log i\rfloor \leq  m\Log m$, for~$m>0$.
Hence:
\[
 |y| = \sum_{i=1}^{m}|\bin{i}\$| = \sum_{i=1}^m\left(\lfloor\Log i\rfloor+2\right)\leq 2m+m\Log m\,.
\]
\qed
\end{proof}

\noindent
Before introducing the announced hierarchy, we prove that given a \pda\ accepting a language~$L\subseteq\Sigma^*$ in~$O(t(n))$ turns, 
it is always possible to build a \pda\ accepting another language using~$O(t(\log n))$ turns.
More precisely, we consider the \emph{extended version of a language~$L$}, denoted as~$\Ext{L}$, which is the language over the alphabet~$\Delta=\Sigma\cup\{0,1,\$\}$ defined as follows:\footnote{%
Notice that we do not require~$\Sigma\cap\{0,1,\$\}=\emptyset$. This allows to iteratively define extended versions of extended versions, without enlarging the input alphabet.
}%
\begin{eqnarray*}
\Ext{L} =\{x_1\$x_2\$\cdots\$x_m\$\$x&\mid&\mbox{$x_i\in\{0,1\}^+, i=1,\ldots,m$, $x\in\Delta^*$,  } \\
			& &\mbox{and at least one of the following holds:}\\
			& &\mbox{\rm(a)~} x_1\$x_2\$\cdots\$x_m\$\notin\ListBin,\\
			& &\mbox{\rm(b)~} |x_m|<|x|,\\
			& &\mbox{\rm(c)~} \,x\in L~\}.
\end{eqnarray*}

\noindent
We point out that, in the above definition, if~$x_1\$x_2\$\cdots\$x_m\$\notin\ListBin$, i.e., condition~(a) is satisfied, 
then the suffix~$x$ could be any string in~$\Delta^*$.
Otherwise, if~$|x_m|\geq|x|$ (i.e., condition~(b) is not satisfied), then~$x$ must belong to~$L$
(condition~(c)).
Strings satisfying only condition~(c) will be crucial to obtain the next results.

\begin{theorem}
\label{th:ext}
	Given a language~$L$ accepted by a \pda~$\ma{M}$ in~$t(n)$ turns,  there exists a \pda~$\ma{M}'$ accepting~$\Ext{L}$  using~$t'(n)$ turns,
	with~$t'(n)\leq\max(1,t(\Log(n)))$.
	Furthermore, if~$\ma{M}$ is a \oca\ then also~$\ma{M}'$ is a \oca.
\end{theorem}
\begin{proof}
From the \pda~$\ma{M}$ accepting~$L$, we define a \pda~$\ma{M}'$ for~$\Ext{L}$ that
given an input~$w$ of the form $x_1\$x_2\$\cdots\$x_m\$\$x$, 
nondeterministically selects one among conditions~(a), (b), (c) in the definition of~$\Ext{L}$, and checks it.

Conditions~(a) and~(b) can be tested using the pushdown store as a counter and making at most one turn.
Hence, when one of these conditions is satisfied by~$w$, one turn is enough.
Condition~(c) is checked by simulating the given \pda~$\ma{M}$ on the suffix~$x$, using at most~$t(|x|)$ turns.
When only~(c) is satisfied by~$w$ then, by Lemma~\ref{lemma:Lbin},~$|x|\leq |x_m|\leq\Log|x_1\$x_2\$\cdots\$x_m\$|\leq\Log|w|$.
So, $t(\Log|w|)$ turns are enough in this case.

Since to check conditions~(a) and~(b) the pushdown store is used as a counter, we easily conclude that if~$\ma{M}$
is a \oca\ then also~$\ma{M}'$ is a \oca.
\qed
\end{proof}

Let us now consider~$(\eq\$)^+$. It is easy to observe that this language is accepted by a \oca\ which uses a linear number of turns.  
Hence, according to Theorem~\ref{th:ext}, 
the language~$\ExtK{k}{(\eq\$)^+}$ is accepted by a \oca\ in~$\max(1,\logk{k}n)$ turns, for each~$k\geq 0$.

We now slightly modify languages~$\ExtK{k}{(\eq\$)^+}$ in order to obtain the announced hierarchy.
More precisely, for each~$k\geq 0$, we define the following language over~$\Sigma=\{a,b,0,1,\$\}$:\footnote{%
	In order to make the technical presentation easier, the languages~$\Lk{k}$ presented here are different from those in the preliminary version of the paper~\cite{Pig25}.
	However, the arguments and the techniques are similar.}
\begin{eqnarray*}
\Lk{k}=
	\{y_k\$\cdots\$y_1\$y_0&\mid& {y_i\in\{\{0,1\}^+\$\}^+, i=1,\ldots,m, y_0\in\{a,b,\$\}^*,} \\
			&&\mbox{and at least one of the following holds:}\\
			&&\mbox{~~\rm(a)~} \exists i, 1\leq i\leq k, y_i\notin\ListBin,\\ 
			&&\mbox{~~\rm(b)~} \exists i, 1\leq i\leq k, |\last{y_i}|\leq\occurr{y_{i-1}}{\$},\\ 
			&&\mbox{~~\rm(c)~} y_0\in(\eq\$)^+\}\,.
\end{eqnarray*}

\noindent
The most relevant difference between~$\Lk{k}$ and~$\ExtK{k}{(\eq\$)^+}$ is condition~(b).
With the new version of this condition, we are still able to obtain a~$\logk{k}n$ upper bound for the number of turns that are sufficient to accepts strings in~$\Lk{k}$, but we will also be able to prove that a number of this order is necessary.

We observe that to check condition~(b), the symbols~$a$ and~$b$ in the block~$y_0$ are not relevant. The remaining symbols in~$y_0$ are~$\$$s. The number of their occurrences in each input string~$w$ satisfying \emph{only} condition~(c) coincides with the number of turns used to accept~$w$.

\begin{example}
\label{example:ex1}
Let us consider the following strings:\\
{
\phantom{~~~}$w_1=\, 1\$10\$11\$100\$\$1\$10\$\$ab\$ab\$aabb\$\$$\\
\phantom{~~~}$w_2=\, 1\$10\$11\$100\$101\$110\$111\$1000\$\$1\$10\$11\$\$ab\$ba\$$\\
\phantom{~~~}$w_3=\, 1\$10\$11\$100\$101\$110\$111\$1000\$\$1\$10\$11\$\$aabb\$$\\
}
All of them are in the language~$\Lk{2}$ and do not satisfy condition~(a) in the definition of the language.
In particular:
\begin{itemize}
\item The string~$w_1$ satisfies condition~(b) because~$|\last{1\$10\$}|=2<\occurr{ab\$ab\$abba\$}{\$}=3$. Furthermore, being~$ab\$ab\$aabb\$\in(\eq\$)^+$, $w_1$ satisfies also condition~(c).
\item The string~$w_2$ satisfies condition~(b) because~$|\last{1\$10\$11\$}| = 2 = \occurr{ab\$ba\$}{\$}$, but it does not satisfy condition~(c).
\item The string~$w_3$ satisfies only condition~(c). In fact, $aabb\$\in(\eq\$)^+$, while~$|\last{1\$10\$11\$100\$101\$110\$111\$1000\$}| = 4 > \occurr{1\$10\$11\$}{\$}=3$
and~$|\last{1\$10\$11\$}| = 2 > \occurr{aabb\$}{\$}=1$,
\qed
\end{itemize}
\end{example}
By refining the argument used to prove Theorem~\ref{th:ext}, we now obtain an upper bound for the number of turns that are sufficient to accept~$\Lk{k}$:

\begin{theorem}
\label{th:loglog-up}
	For each~$k\geq 0$ there exists a \oca~$\ma{A}_k$ accepting the language~$\Lk{k}$ in~$max(1,\logk{k}n)$ turns.
\end{theorem}
\begin{proof}
Fixed~$k\geq 0$, let~$w=y_k\$\cdots\$y_1\$y_0$ be a string in~$\Lk{k}$, according to the above definition.
Each one of the conditions~(a) and~(b) can be verified using only one turn.
So a greater number of turns could be necessary for strings satisfying \emph{only} condition~(c).

In particular, if~$w$ satisfies only~(c), then it is accepted in~$t$ turns, where~$t=\occurr{y_0}{\$}$.
We now relate~$t$ to the length~$n$ of~$w$.
From the negation of~(a) and Lemma~\ref{lemma:Lbin}, we obtain that, for $i=1,\ldots,k$:
\[\last{y_i}\geq 2^{|\last{y_i}|-1}\,.\] 
Furthermore, when~$i>1$, from the negation of~(b) we obtain that~$|\last{y_i}| > \occurr{y_{i-1}}{\$}$, i.e., $|\last{y_i}|-1 \geq \occurr{y_{i-1}}{\$}$, so obtaining
\[
\last{y_i}\geq 2^{|\last{y_i}|-1}\geq 2^{\occurr{y_{i-1}}{\$}}=2^{\last{y_{i-1}}}\,
\]
and, hence~$\Log\last{y_i}\geq\last{y_{i-1}}$, for~$i=2,\ldots,n$.
For the same reason, $\last{y_1}\geq 2^{\occurr{y_0}{\$}}=2^t$ and $\lg\last{y_1}\geq t$.
This allows to obtain that~$t\leq\logk{i}\last{y_i}$, for~$i=1,\ldots,k$. Since~$n=|w|\geq\last{y_k}$, we finally obtain that~$t\leq\logk{k}n$.
This completes the proof.
%
\qed
\end{proof}

We are now going to obtain a lower bound for the number of turns necessary to accept~$\Lk{k}$.
To this aim, we now introduce some strings that, as we will prove, need a number of
turns of the order of $\logk{k}n$ to be accepted.

First, for integers~$t,k\geq 0$, let us consider the following numbers:
\[
n_{t,k}=\left\{\begin{array}{ll}
	t & \mbox{if } k=0,\\
	2^{n_{t,k-1}} & \mbox{otherwise}.
\end{array}\right.
\]
For~$k>0$, we define the strings:
\[
y_{t,k}=\bin{1}\$\bin{2}\$\cdots\$\bin{n_{t,k}}\$
\]
and
\[
w_{t,k}=y_{t,k}\$y_{t,k-1}\$\cdots\$y_{t,1}\$
\]
\begin{lemma}
\label{lemma:wtkLenght}
If~$t>2$ then $|w_{t,k}| < n_{t,k}^2$.
\end{lemma}
\begin{proof}
To prove the result, we make use of the following inequality, which holds for each~$x>2$ and can be easily verified:
\begin{equation}
\label{eq:22x}
2^x(2+x)+x^2+1 <  2^{2x}\,.
\end{equation}
We observe that, by Lemma~\ref{lemma:Lbin}, $|y_{t,k}|\leq 2n_{t,k}+n_{t,k}\Log n_{t,k}$, for~$k>0$.

We prove the claimed upper bound on~$|w_{t,k}|$ by induction on~$k$.

If~$k=1$, using~(\ref{eq:22x}) with~$x=t$ and~$n_{t,1}=2^t$, we obtain:
\[|w_{t,1}|=|y_{t,1}|+1\leq 2^{t+1}+2^tt+1=2^t(2+t)+1<2^{2t} = n_{t,1}^2\,.\]
For~$k>1$, using~$|w_{t,k-1}|<n^2_{t,k-1}$ as induction hypotesis and~(\ref{eq:22x}), we obtain:
\begin{eqnarray*}
|w_{t,k}|&=&|y_{t,k}|+1+|w_{t,k-1}| \\
		&<& 2n_{t,k}+n_{t,k}\Log n_{t,k}+1+ n_{t,k-1}^2\\
		& = & 2\cdot 2^{n_{t,k-1}}+2^{n_{t,k-1}}\cdot n_{t,k-1} + 1 + n_{t,k-1}^2\\		
		& = & 2^{n_{t,k-1}}(2+ n_{t,k-1}) + n_{t,k-1}^2 + 1\\
		& < & 2^{2n_{t,k-1}} = (2^{n_{t,k-1}})^2 = n_{t,k}^2\,.
\end{eqnarray*}
\qed
\end{proof}

\begin{lemma}
\label{lemma:wtk}
	For each integer~$h>0$ and each string~$x\in(\eq\$)^h$, the string~$w=w_{t,k}x$ belongs to the language~$\Lk{k}$.
	Furthermore, $w$ satisfies \emph{only condition~(c)} in the definition of~$\Lk{k}$ if and only if~$h\leq t$.
\end{lemma}
\begin{proof}
	By definition, the string~$w$ cannot satisfy condition~(a).
	
	Concerning condition~(b), we observe that, by definition of~$w_{t,k}$, if~$i>1$ then~$|\last{y_i}|=\lfloor\Log n_{t,i}\rfloor + 1 = n_{t,i-1} + 1$
	and~$\occurr{y_{i-1}}{\$}=n_{t,i-1}$, that gives~$|\last{y_i}| > \occurr{y_{i-1}}{\$}$.
	So, condition~(b) is satisfied if and only if~$|\last{y_1}|\leq\occurr{y_0}{\$}$.
	Since~$\last{y_1}=n_{t,1}=2^t$, $|\last{y_1}|=t+1$, while~$\occurr{y_0}{\$}=h$, we obtain that condition~(b) is satisfied if and only if~$t+1\leq h$, i.e., $h>t$.
		
	This allows us to conclude that~$w$ satisfies \emph{only condition~(c)} if and only if~$h\leq t$.
\qed
\end{proof}

We are now able to prove the lower bound:

\begin{theorem}
\label{th:loglog-lb}
	For each \pda~$\ma{M}_k$ accepting the language~$\Lk{k}$, $k>0$, and
	for infinitely many integers~$n$, there is a string of length~$n$ requiring at 
	least~$\logk{k}n-2$ turns to be accepted by~$\ma{M}_k$. 
\end{theorem}
\begin{proof}
Fixed an integer~$k>0$, let us consider a \pda~$\ma{M}_k$ accepting~$\Lk{k}$. We want to show that~$\ma{M}_k$
needs a number of turns growing as~$\logk{k}n$ to accept inputs of length~$n$.
Let~$N$ be the constant of Lemma~\ref{lemma:lowerBoundEQ} for~$\ma{M}_k$.

Given~$t>2$, let us take~$N'=N\cdot(|w_{t,k}|+t+1)$, and consider the strings~$x=(a^{N'}b^{N'}\$)^t$ and~$w=w_{t,k}x$.
According to Lemma~\ref{lemma:wtk}, $w\in\Lk{k}$ and it satisfies \emph{only condition~(c)} in the definition of~$\Lk{k}$.
We now prove that~$t$ turns are necessary to accept~$w$.
To this aim, we observe that~$w=w_{t,k}z_1\$z_2\cdots z_t\$$ with~$z_i=a^{N'}b^{N'}$, $i=1,\ldots,t$.
Since~$w$ satisfies only condition~(c), if we replace in it some of the blocks~$z_i$ with some string in~$\{a,b\}^*\setminus\eq$, then we would obtain
a string not in~$\Lk{k}$, while if the replacements are made only using strings in~$\eq$, then the resulting string belongs to~$\Lk{k}$.
In other words, $w_{t,k}\{a,b\}^*\$\{a,b\}^*\cdots \{a,b\}^*\$\cap \Lk{k}=w_{t,k}\eq\$\eq\cdots\eq\$$,
Hence, from Lemma~\ref{lemma:lowerBoundEQ}, it follows that~$t$ turns are necessary to accept~$w$.

To complete the proof, we state a relationship between the length~$n$ of~$w$ and the number~$t$ of turns required to accept it:
\begin{eqnarray*}
	n &=& |w| = |w_{t,k}| + |x| = |w_{t,k}| + (2N'+ 1) t \\
	  &=& |w_{t,k}| + (2N(|w_{t,k}|+t+1)+1)t \\
	  &=& |w_{t,k}| \cdot \left( 1 + \left(2N + \frac{2N(t+1)+1}{|w_{t,k}|} \right) \cdot t\right)
\end{eqnarray*}
For~$k$ sufficiently large, $|w_{t,k}| > 1+\left(2N + \frac{2N(t+1)+1}{|w_{t,k}|} \right) \cdot t$.
Hence, from Lemma~\ref{lemma:logk-prod}, we obtain that:
\[\logk{k}n < 1+ \logk{k}(|w_{t,k}|) \leq 1 + \logk{k} n_{t,k}^2 \leq 2 + \logk{k} n_{t,k} = 2 + t\,.
\]
Hence, $t > \logk{k}n - 2$.

This allows to conclude that the number of turns necessary to accept strings of length~$n$ is at least~$\logk{k}n-2$, for infinitely many integers~$n$. 
\qed
\end{proof}

By combining Theorem~\ref{th:loglog-up} and Theorem~\ref{th:loglog-lb} we obtain:

\begin{theorem}\label{th:hierarchy}
	For each~$k\geq 0$ there exists a language which is accepted by a \oca\ in at most~$max(1,\logk{k}n)$ turns 
	but that cannot be accepted by a \pda\ in~$o(\log^{(k)}n)$ turns.
\end{theorem}
\begin{proof}
Given~$k\geq 0$, let us consider the language~$\Lk{k}$. By Theorem~\ref{th:loglog-up}, $max(1,\logk{k}n)$ turns are enough to recognize it.
Furthermore, $\Lk{0}=(\eq\$)^+$ cannot be accepted in~$o(n)$ turns and,
by Theorem~\ref{th:loglog-lb}, for~$k>0$  the language~$\Lk{k}$ requires at least~$\logk{k}n-2$ turns to be accepted by a \pda.
Since this function grows as~$\logk{k}n$, we can conclude that~$\Lk{k}$ cannot be accepted in~$o(\log^{(k)}n)$ turns.
\qed
\end{proof}

As a consequence of Theorem~\ref{th:hierarchy}, we obtain an infinite proper hierarchy with respect to the number of turns:
for each~$k\geq 0$ the class of languages accepted in at most~$\log^{(k)}n$ turns properly includes that of languages accepted in at most~$\log^{(k+1)}n$ turns.

\section{Even Much Less Turns}
\label{sec:ustar}

In this section, we present a language~$\Ustar$ that can be accepted with less turns than each language in the 
hierarchy discussed in Section~\ref{sec:hie}. Indeed, we  prove that~$\logstar n$ turns are sufficient to recognize~$\Ustar$. We also prove that
this number is necessary, for infinitely many~$n$.\footnote{%
	Even in this case, the language~$\Ustar$ considered here is slightly different than the language presented in the preliminary version of the paper~\cite{Pig25}.
	For the `new' language~$\Ustar$, we are able to prove that~$\logstar n$ turns are necessary (Theorem~\ref{th:lstar-lower}), while for the `old' one we
	were only able to prove that an unbounded number of turns is necessary.}

To introduce the language~$\Ustar$, we preliminarily consider, for each~$k>0$, the following modification of the language~$\Lk{k}$:
\[
\begin{array}{ll}
\Uk{k}=\{y_kz_ky_{k-1}\cdots z_2y_1z_1\mid&y_i\in(\{0,1\}^+\$)^+, z_i\in\{a,b\}^+,  i=1,\ldots,k,\\
			&\mbox{and at least one of the following holds:}\\
			&\mbox{~~\rm(a)~} \exists i, 1\leq i\leq k, y_i\notin\ListBin,\\ 
			&\mbox{~~\rm(b)~} \exists i, 1< i\leq k, |\last{y_i}| \leq \occurr{y_{i-1}}{\$},\\ 
			&\mbox{~~\rm(c)~} \forall i, 1\leq i\leq k, z_i\in\eq\,\}\,.
\end{array}
\]
We observe that each string~$w=y_k\$y_{k-1}\$\cdots\$y_1\$y_0\in\Lk{k}$,
with~$y_0\in(\{a,b\}^+\$)^k$, can be transformed into a string of~$\Uk{k}$
by replacing the delimiters between consecutive blocks~$y_i$'s by maximal factors of~$y_0$ in~$\{a,b\}^+$.
The conditions~(a), (b), and~(c) are similar to those defining~$\Lk{k}$. Hence, the recognition of~$\Uk{k}$ can be done in a similar way,
using~$k$ turns.

\begin{example}
\label{example:ex2}
As discussed below, strings satisfying only condition~(c) in the definition of~$\Uk{k}$ play a crucial role in the study of the number of turns.
We present an example of such a string, for~$k=4$.
First let us consider the following strings:\\
{
\phantom{~~~}$y_1=\, 1\$$\\
\phantom{~~~}$y_2=\, 1\$10\$$\\
\phantom{~~~}$y_3=\, 1\$10\$11\$100\$$\\
\phantom{~~~}$y_4=\, 1\$10\$11\$100\$101\$110\$111\$1000\$1001\$1010\$1011\$1100\$1101\$1110\$1111\$10000\$$\\
}
We point out that~$y_i=\bin{1}\$\bin{1}\$\cdots\$\bin{\tetra{i-1}}\$$, $i=1,\ldots,4$. 
Let us take a string~$u=y_4z_4y_3z_3y_2z_2y_1z_1$, where~$z_1,z_2,z_3,z_4\in\{a,b\}^+$.
Then:
\begin{itemize}
\item Since for~$i=1,\ldots,4$, $y_i\in\ListBin$, $u$ does not satisfy condition~(a).
Furthermore~$\occurr{y_i}{\$}=\last{y_i}=\tetra{i-1}$ and~$|\last{y_i}|=\Log(\tetra{i-1}) + 1 > \Log(\tetra{i-1})$.
\item From the previous equalities, it turns out that for~$i=2,3,4$, $|\last{y_i}|>\tetra{i-2}=\last{y_{i-1}}=\occurr{y_{i-1}}{\$}$.
Hence, condition~(b) is not satisfied.
\item This allows to conclude that the string~$u$ belongs to~$\Uk{4}$ if and only if condition~(c) is satisfied, namely if and only if all
the strings~$z_i$ are in~$\eq$, with~$i=1,\ldots,4$.
\item In particular, if we choose~$z_1=z_2=z_3=z_4=ab$, then the string~$u$ so obtained is the shortest string in~$\Uk{4}$ satisfying only~(c).
This idea will be expanded in the proofs of the next results, in particular in the proof of Theorem~\ref{th:lstar-lower}.\qed
\end{itemize}
\end{example}

\noindent
We now define the language we are interested in as $\Ustar = \bigcup_{k\geq 0}\Uk{k}$.

\begin{theorem}
	\label{th:lstar-up}
	The language~$\Ustar$ is accepted by a \oca\ in~$\logstar n$ turns.
\end{theorem}
\begin{proof}
	Given an  input string~$w=y_kz_ky_{k-1}\cdots z_2y_1z_1$, $k>0$, in the appropriate form (this can
	be verified while scanning the input, using the finite control), in order to check if~$w\in\Ustar$, a \oca~$\ma{B}$ 
	can guess one condition among~(a), (b), and~(c).
	Conditions~(a) and~(b) are verified nondeterministically choosing one block~$y_i$ and using one turn, as in the
	proof of Theorem~\ref{th:loglog-up}.
	To test condition~(c), $\ma{B}$ simulates the \oca~$\meq$ on each block~$z_i$, so making exactly~$k$ turns.
	This allows us to conclude that the maximum number of turns~$k$ is reached when~$w$ satisfies \emph{only} condition~(c).
	
	We are now going to estimate how large can be~$k$, with respect to the length~$n$ of an input~$w$ satisfying
	only~(c).
	
	For~$i=1,\ldots,k$, since condition~(a) is not satisfied, using Lemma~\ref{lemma:Lbin} we obtain that~$\last{y_i}\geq 2^{|\last{y_i}|-1}$.
	Since condition~(b) is not satisfied, when~$i>1$, we also have that~$|\last{y_i}| > \occurr{y_{i-1}}{\$} =\last{y_{i-1}} $, i.e., $|\last{y_i}|-1 \geq \last{y_{i-1}}$.
	Hence, for~$i>1$, $\last{y_i} \geq 2^\last{y_{i-1}}$. By definition, $\last{y_1}$ should be at least~$1=\tetra{0}$.
	This allows to conclude that, for~$i=1,\ldots,k$, $\last{y_i}\geq\tetra{i-1}$ and, hence, $|y_i|>\tetra{i-1}$.
	
	Hence, $n=|w|>|y_k|>\tetra{k-1}$, that gives~$\logstar n> k-1$ and, being~$\logstar n$ an integer, $\logstar n\geq k$.
	We remind the reader that~$k$ turns are sufficient to accept~$w$. 
	
	This allows to conclude that~$\logstar n$ turns are enough to recognize strings of length~$n$ in~$\Ustar$.
\qed
\end{proof}

We conclude this section by proving a matching lower bound:

\begin{theorem}
	\label{th:lstar-lower}
	Let~$\ma{M}$ be a \pda\ accepting the language~$\Ustar$. Then, for infinitely many~$n$ there exists a string of length~$n$ that requires~$\logstar n$ turns to be accepted by~$\ma{M}$.
\end{theorem}
\begin{proof}
	For each~$i>0$, let us consider~$y_i=\bin{1}\$\bin{1}\$\cdots\$\bin{\tetra{i-1}}\$$.
	By Lemma~\ref{lemma:Lbin},
	\[|y_i|\leq 2\cdot\tetra{i-1} + \tetra{i-1}\Log(\tetra{i-1}) = \tetra{i-1}(2+\Log(\tetra{i-1}))\,.\]
	Let~$N$ be the constant of Lemma~\ref{lemma:lowerBoundEQ} for~$\ma{M}$.
	For~$k>0$, let~$N_k=N\cdot k\cdot|y_k|$ and consider the string
	\[
	u_k= y_ka^{N_k}b^{N_k}y_{k-1}a^{N_k}b^{N_k}\cdots y_1a^{N_k}b^{N_k}\,.
	\]
	It is immediate to observe that the string~$u_k$ does not satisfy condition~(a) in the definition of~$\Ustar$.
	Moreover, for~$i=2,\ldots,k$:
\[
|\last{y_i}|=|\tetra{i-1}|=\Log\tetra{i-1}+1=\tetra{i-2}+1=\last{y_{i-1}}+1>\occurr{y_{i-1}}{\$}\,.
\]	
	Hence, the string~$u_k$ satisfies only condition~(c) in the definition of~$\Ustar$.
	Furthermore, if~$k$ is sufficiently large, then~$N_k\geq N\cdot(|y_k\cdots y_1|+1)$ and
	by replacing some of the blocks~$a^{N_k}b^{N_k}$ in~$u_k$ with strings in~$\{a,b\}^*$, we obtain strings in~$\Ustar$ only if strings
	from~$\eq$ are used in the replacements. From Lemma~\ref{lemma:lowerBoundEQ}, we derive that~$k$ turns are necessary to accept~$u_k$
	
	To complete the proof, we are going to relate~$k$ to the length~$n$ of~$u_k$.
	Using Lemma~\ref{lemma:Lbin}, we observe that:
	\begin{eqnarray*}
		n = |u_k| & \leq & k\cdot|y_ka^{N_k}b^{N_k}| \\
		& = & k\cdot(|y_k|+2Nk\cdot|y_k|) = |y_k|\cdot (1+2Nk)\cdot k\\
		& \leq & \tetra{k-1} \cdot (2+\Log(\tetra{k-1})) \cdot (1+2Nk)\cdot k\,.
	\end{eqnarray*}
	If~$k$ is large enough, then we have that~$\tetra{k-1} > (2+\Log(\tetra{k-1})) \cdot (1+2Nk)\cdot k$.
	Using Lemma~\ref{lemma:logk-prod} we obtain:
	\[
		\logk{k-1}n = \logk{k-1}|u_k|  \leq  1 + \logk{k-1}(\tetra{k-1}) = 2\,.
	\]
	Hence, $\logk{k}n \leq 1$ that gives~$k\geq\logstar n$.
\qed
\end{proof}

\section{Concluding Remarks}
\label{sec:conclu}

In some recent papers on pushdown and one-counter automata, we considered two complexity measures: the \emph{height of the pushdown} and
the \emph{push complexity}~\cite{Pig2026,PP25,PP23}. In this paper, we continued that line of research, 
investigating~\emph{turn complexity}.

While there are immediate relationships between the first two measures,
it is not clear if turn complexity can be in some extent related to the other two measures.
The examples of \pdas\ using minimal non-constant~$O(\log\log n)$ height or push complexity presented in the above-cited
papers make at most one turn in each accepting computation. It would be interesting to know if there exist machines using a `small' 
but unbounded number of turns and, at the same time, $O(\log\log n)$ height or push complexity, or if the use of an unbounded
number of turns implies a larger use of height and push operations.
We point out that the \oca~$\ma{A}_k$ accepting the language~$\Lk{k}$ and described in the proof of Theorem~\ref{th:loglog-up},
can use an unbounded height to recognize strings satisfying only condition~(c).
It would be interesting to investigate a complexity measure that keeps into account at the same time the height (or the number of push operations) and the number to turns.

\smallskip

Another question concerns the possibility of reformulating our results in terms of grammars.
It is know that finite-turn pushdown automata corresponds to ultralinear and to non-terminal bounded grammars~\cite{GS66}.
Suppose to have a \pda~$\ma{M}$ accepting in a number of turns bounded by a non-constant function~$f(n)$. Can we relate~$f(n)$
to some measure on a context-free grammar obtained from~$\ma{M}$ using a `standard' transformation?

\smallskip

Finally, we present the problem of the restrictions to the case of unary languages, i.e., languages defined over a one-letter alphabet, or the  case of bounded languages. We expect that in these cases all the properties investigated in Section~\ref{sec:undec} become decidable.
Furthermore, the techniques used in Sections~\ref{sec:sublinear}, \ref{sec:hie}, and~\ref{sec:ustar} to reduce the number of turns do not seem to work in these cases.
We conjecture that it is not be possible to have a number of turns which is not constant, but arbitrarily slowly increasing, as in the general case.

\bibliographystyle{alpha}
\bibliography{biblio}

\end{document}